\documentclass{CSML}
\pdfoutput=1

\usepackage{lastpage}
\newcommand{\dOi}{14(1:23)2018}
\lmcsheading{}{1--\pageref{LastPage}}{}{}%
{Jul.~21,~2016}{Mar.~20, 2018}{}

\usepackage{hyperref}
\hypersetup{hidelinks}

\SetLabelAlign{parright}{\parbox[t]{\labelwidth}{\raggedleft#1}}

\usepackage{amsmath,amsfonts} 
\usepackage{stmaryrd}

\usepackage[applemac]{inputenc}

\newcommand{\N}{{\mathbb N}}

\newcommand{\F}{{\bf F}}
\newcommand{\I}{{\bf I}}

\newcommand{\PCF}{\mathsf{PCF}}
\newcommand{\LCF}{\mathsf{LCF}}

\begin{document}
\title{The sequential functionals of type $(\iota \rightarrow \iota)^n \rightarrow \iota$ form a \emph{dcpo} for all $n \in \N $}
\author{Dag Normann} \address{Department of Mathematics, The University of Oslo, P.O. Box 1053, Blindern N-0316 Oslo, Norway}
\email {dnormann@math.uio.no}
\subjclass{[{\bf Theory of Computation}];Models of computation-Computability-Lambda calculus}
\begin{abstract}\noindent We prove that the sequential functionals of some fixed types at type level 2, taking finite sequences of unary functions as arguments,  form a directed complete partial ordering. This gives a full characterisation of the types for which  the partially ordered set of sequential functionals has this property.

As a tool, we prove a normal form theorem for the finite sequential functionals of the types in question.
\end{abstract}
\maketitle
\section{Introduction}
When Scott \cite{Scott} introduced his calculus $\LCF$ and the model of partial continuous functionals for it, it was clear that this model is not fully abstract. There will be finite elements in this model, e.g. parallel or, that are not $\LCF$-definable, and consequently there will be observationally equivalent terms in $\LCF$ that will be interpreted as different objects in the Scott model.

Milner \cite{milner} constructed an alternative model for $\LCF$ based on Scott domains, a model that is fully abstract, and proved that up to isomorphism, there is only one such model. Since Milner's model has the cardinality of the continuum, there will be a lot of elements in this model that are not $\LCF$-definable. However, if we let $\LCF^\Omega$ denote the calculus where we extend $\LCF$ with one constant  $\hat f$ for each $f:\N \rightarrow \N$, together with the rewrite rules $\hat fa \rightarrow b$ whenever $f(a) = b$, the question  whether all objects in Milner's model are $\LCF^\Omega$-definable remained open for many years.

We let the \emph{sequential functionals} be the hereditarily minimal extensional model for $\LCF^\Omega$. In this paper, we will investigate the closure properties of the sequential functionals of types $(\iota \rightarrow \iota)^n \rightarrow \iota$, where $\iota$ denotes the base type interpreted as the flat domain $\N_\bot$. We will use one of the known characterisations of the sequential functionals, representing them  as the hereditarily extensional   collapse of the sequential procedures (see Section 2). The main result is that the sequential functionals of the types $(\iota \rightarrow \iota)^n \rightarrow \iota$ will form a directed complete partial ordering. As a tool, an extended \emph{normal form theorem} will be obtained for these types.

The first systematic  analysis of the $\LCF$-definable functionals, seen as a subclass of the Scott model,  was due to Sazonov \cite{Sa}. To this end, he introduced the concept of a \emph{strategy}, a concept that is close to sequential procedures as we define them. Later characterisations are due to Nickau \cite{Nickau}, to Abramsky, Jagadeesan and Malacaria \cite{AJM} and to Hyland and Ong \cite{HO}. These characterisations appeared during the 1990's. The question whether they characterised Milner's model as well was, however, left open through that decade. This question  was answered in Normann \cite{DN1}.  There it is proved that the sequential functionals of pure type 3 do not form a \emph{dcpo}, a directed complete partial ordering, and thus do not coincide with Milner's interpretation of this type. Normann and Sazonov \cite{N-S} studied some anomalies inspired by \cite{DN1}, and among other things, they proved that even the sequential functionals of type $(\iota,\iota \rightarrow \iota) \rightarrow \iota$ do not form a \emph{dcpo}. This type is of level 2 in the crude stratification of the finite types into levels. However, in \cite{N-S} the authors were unable to offer a full characterisation of when the sequential functionals of a given type form a \emph{dcpo}, and it is this gap of knowledge that will be closed in this note.

\smallskip
The sequential functionals $F$ of type $(\iota,\iota \rightarrow \iota)\rightarrow \iota$ will take sequential functions $f:\N_\bot \times \N_\bot \rightarrow \N_\bot$ as arguments and give elements of $\N_\bot$ as values. It is essential for the construction in \cite{N-S} that $f$ ranges over the \emph{binary} functions. We do not need the definition of sequential functionals of this type in this note. 

\section{Sequential procedures and functionals}
Throughout this note, we will let $n > 0$ be a fixed  integer.
We let the set of monotone functions $f:\N_\bot \rightarrow \N_\bot$ be our interpretation of the type $\iota \rightarrow \iota$, and we let $\sqsubseteq$ be the standard pointwise ordering on this set. We will use $f$, $g$ etc. for such functions. We will use $\vec f$ for $n$-sequences $(f_1 , \ldots , f_n)$ of functions.

\smallskip
We will interpret $\sigma = (\iota \rightarrow \iota)^n \rightarrow \iota$ as the set of \emph{sequential functionals} of the type.  They can be given as the denotational, hereditarily extensional,  interpretations of the \emph{ sequential procedures} as defined below. The sequential procedures will be infinite syntax trees, infinitary terms of type $\iota$, where we essentially follow the notation from \cite{L-N}. All of our procedures will be terms with free variables among $x_1, \ldots , x_n$, a list of variables of type $(\iota \rightarrow \iota)$.
\begin{rem} Since the interpretation $\N_\bot$ of $\iota$ is a sequential retract of the interpretation $\N_\bot \rightarrow \N_\bot$ of $(\iota \rightarrow \iota)$, we will indirectly also have covered the cases $\sigma_1 , \ldots , \sigma_n \rightarrow \iota$ where each $\sigma_i$ is either $\iota$ or $(\iota \rightarrow \iota)$.
When the meaning is clear, we will drop the brackets, and write $\iota \rightarrow \iota$.\end{rem}
\begin{defi}\label{2.1} For $a \in \N_\bot$, we let $a$ also be a constant term denoting itself.
The \emph{sequential procedures} in the variables $x_1 , \ldots , x_n$ of type $\iota \rightarrow \iota$ will be the largest class of infinitary terms $P$ of type $\iota$ such that either  $P = c$ for some $c \in \N_\bot$ or $P$ is of the form
\begin{center}$P$ =  case $x_i(Q)$ of $\{a \Rightarrow P_a\}_{a \in \N}$\end{center} where  $Q$ and each $P_a$ are sequential procedures. 
We will drop the index $a \in \N$ in the notation in this paper. \end{defi}
\begin{rem}This definition is a special, and slightly simplified, case of the more general definition of the sequential procedures of type $\sigma$ for any type $\sigma$, as given in \cite{L-N}. We have actually defined the relevant class of \emph{expressions} as defined in \cite{L-N}, but since we do not need  terms defined by $\lambda$-abstractions here, we name them \emph{procedures}.

The sequential  procedures were used as the basis for investigating the sequential functionals both in \cite{N-S} and in \cite{L-N}. In both, there are proofs of the fact that the sequential procedures form an applicative structure. The proofs differ in some essential ways, and in both cases they are non-trivial.\end{rem}
The class of sequential procedures is given by a coinductive definition. We may of course view the clauses as clauses in an inductive definition. If we do so, and in addition restrict $a$ to finite sets $A \subseteq \N$ instead of $\N$ (meaning, in effect, that $P_a = \bot$ if $a \not \in A$), we get the \emph{finite sequential procedures}. 
\begin{defi}\label{2.2} 
Each procedure $P$ with variables among $x_1 , \ldots , x_n$ will define a \emph{ sequential functional} $ \llbracket P\rrbracket{}:(\N_\bot \rightarrow \N_\bot)^n \rightarrow \N_\bot$ by the \emph{inductive} definition of $\llbracket P\rrbracket{}(\vec f)$:

\begin{enumerate}
\item $\llbracket P\rrbracket{}(\vec f) = a$ when $P = a$.
\item If \begin{center}$P$ =  case $x_i(Q)$ of $\{a \Rightarrow P_a\}_{a \in \N}$,\end{center} we let $\llbracket P\rrbracket{}(\vec f) = b$ if for some $a\in \N$  and  for some $c \in \N_\bot$ we have that $\llbracket Q\rrbracket{}(\vec f) = c$, $f_i(c) = a$, and $\llbracket P_a\rrbracket{}(\vec f) = b$.
\item$\llbracket P\rrbracket{}(\vec f) = \bot$ if no value is given through (1) or (2). \end{enumerate}\end{defi}
\noindent We write $\llbracket P\rrbracket{}(\vec f)\!\!\downarrow$ if there is some $b \in \N$ such that $\llbracket P\rrbracket{}(\vec f) = b$.
There may be three reasons for why $\llbracket P\rrbracket{}(\vec f) = \bot$. One is that the syntax tree of $P$ is not well founded, and that we follow an infinite branch in this syntax tree when we evaluate $\llbracket P\rrbracket{}(\vec f)$. Another is that  we, during the evaluation, are led to some $a \Rightarrow \bot$. The third reason is that we need to evaluate some $f_i(\llbracket Q\rrbracket{}(\vec f))$ where the value of $f_i(\llbracket Q\rrbracket{}(\vec f)) = \bot$. There is a qualitative difference here: in the first case the process does not terminate, while in the other cases the process terminates with the value $\bot$. In our interpretation $\llbracket P\rrbracket{}$ we do not distinguish between these  cases.
\begin{defi}\label{2.2a}  Let $k \in \N$, and let $\F_k$ be the set of sequential functionals of our fixed type over the interpretation $\{\bot, 0 , \ldots , k\}$ of $\iota$. We let $\F$ denote the full set of sequential functionals of the type in question.
For each $k \in \N$, let $(\nu_{k},\pi_{k})$ be the standard sequential embedding-projection pair between $\F_k$ and $\F$ induced by the inclusion map and the restriction map at base type.
\end{defi}
The crucial properties are:
\begin{center}
 $\pi_{k} \circ \nu_{k} = id_{\F_k}$ and $\nu_{k} \circ \pi_{k} \sqsubseteq id_{\F}$
\end{center}
Since we will consider $\F_k$ to be a subset of $\F$  we will also consider $\nu_{k}$ to be the inclusion map, and not mention it again. We will need $\pi_{k}$ in the sequel, viewing it as a restriction map. We will use the obvious fact that each $\F_k$ is a finite set.

We say that two procedures are \emph{equivalent} if they define the same functional. In \cite{N-S} a procedure was defined to be \emph{normal} if we will only use $Q$s of the form $c$ with $c \in \N$ in the coinductive definition. It is not the case that all procedures will be equivalent to procedures in normal form, so this term is slightly misleading.
\smallskip

 As usual, we will say that a function $f$ is \emph{strict} if $f(\bot) = \bot$. 
It is the presence of non-strict functions, and the fact that sometimes $\llbracket P\rrbracket{}(\vec f)$ may terminate \emph{because} some $f_i$ is not strict, that makes it impossible to prove a normalisation theorem.  Plotkin  proved that the strict functions between \emph{concrete domains} will form \emph{dcpo}s, see  Kahn and Plotkin \cite{KP} for definitions.
From now on, $S$ will always be a subset of $\{1 , \ldots , n\}$.

\begin{samepage}
\begin{defi}\label{2.3}
 If  $F$ is a sequential functional, we let $F_S$ be defined by
$F_{S}(\vec f) = a$ if there is a sequence $\vec g \sqsubseteq \vec f$ such that 
\begin{itemize}
\item $g_i$ is strict when $i \in S$,
\item $F(\vec g) = a$.
\end{itemize}
\end{defi}
\end{samepage}

\noindent $F_{S}$ will be sequential when $F$ is, since $F_S$ is the composition of $F$ and sequential ``strictification"-functionals. This argument requires that we know that the sequential functionals form an applicative structure, see \cite{L-N} or \cite{N-S}. 
Definition \ref{2.4} will contain a direct construction of a procedure for $F_S$ from a procedure for $F$ in the finite case, and then it will follow that $F_S$ is sequential from our main theorem.

 A finite sequential procedure may be nested in two directions, in depth (or to the left) and in length (or to the right).  A \emph{left bound} of a procedure will be $\geq$ the maximal number of nestings to the left in any branch of the syntax tree.  We will prove that if $F$ is a finite sequential procedure, there is a special procedure for $F_{S}$ with a left bound depending only on $S$, and this will be our alternative to the normal form theorem for sequential procedures restricted to strict arguments. 
\begin{defi}\label{2.4} Let $P$ be a finite, non-empty sequential procedure.
We define the procedure $P[S]$ by cases as follows:
\begin{enumerate}[label=\arabic*.]
\item If $P = c$, we let $P[S] = P$.
\item Let $P$ =  case  $x_i(Q)$ of $ \{a \Rightarrow P_a\}$:
\begin{itemize}[align=parright,labelsep=1.5em]
\item[2.1.] If $i \not \in  S$, let
\begin{center}$P[S]$ = case $x_i(Q[S \cup \{i\}])$ of $\{a \Rightarrow P_a[S]\}$\end{center}
\item[2.2.] If $i \in S$, and $Q = c$ we  let 
\begin{center} $P[S]$ =  case $x_i(c)$ of $ \{a \Rightarrow P_a[S]\}$. \end{center}
\item[2.3.] If $i \in S$ and $Q$ =  case $x_j(R)$ of $ \{b \Rightarrow Q_b\}$, we first let
\begin{center}
$\hat P = $ case $x_j(R)$ of $ \{ b \Rightarrow$ case $ x_i(Q_b)$ of $ \{a \Rightarrow P_a\}\}$
\end{center} and then let $P[S] = \hat P[S]$.
\end{itemize}
\end{enumerate}

\end{defi}
\begin{lem}\label{2.5}{\em
Let $P$ be a finite sequential procedure, $S \subset \{1 , \ldots , n\}$.
\begin{enumerate}[label=\alph*)]
\item $P[S]$ is well defined, i.e. the rewriting terminates.
\item If $F = \llbracket P\rrbracket{}$, then $F_{S} = \llbracket P[S]\rrbracket{}$.
\end{enumerate}
}\end{lem}
\begin{proof}
a) is proved by finding a suitable complexity measure as follows:
When $P$ is a finite sequential procedure, the evaluation of $\llbracket P\rrbracket{}(\vec f)$ will have a bounded finite length, even if we insist on evaluating $\llbracket Q\rrbracket{}(\vec f)$ to a value in $\N_\bot$ when $f_i$ is not strict and $f_i(\llbracket Q\rrbracket{}(\vec f))$ is asked for in the process. For input $\vec f$ we will let the length of the evaluation of $\llbracket P\rrbracket{}(\vec f)$ be the number of evaluations $\llbracket Q\rrbracket{}(\vec f)$ we thus have to perform in the process.  Then we use induction on the maximal possible length of an evaluation. In Case 2.3 $\hat P$ will have lower complexity than $P$ since we omit the need to evaluate $Q(\vec f)$, while both in the evaluation of $\llbracket P\rrbracket{}(\vec f)$ and $\llbracket \hat P\rrbracket{}(\vec f)$ we have to evaluate $\llbracket R\rrbracket{}(\vec f)$, some $\llbracket P_a\rrbracket{}(\vec f)$ and some $Q_b(\vec f)$ (with all their sub-evaluations), though in different orders in the two cases. For the rest of the cases, it is easy to see how the induction hypothesis  is used.

In order to prove b) we will use  induction on the number of rewritings following the clauses of Definition \ref{2.4} we need in order to find $P[S]$, and prove that $\llbracket P[S]\rrbracket{}(\vec f) = \llbracket P\rrbracket{}_{S}(\vec f)$ for each input $\vec f$.
We go through each of the steps:
\begin{itemize}[align=parright,leftmargin=3em,labelsep=2em,midpenalty=99]
\item[1.] This case is trivial

\item[2.1.] There will be two sub-cases:
\newline
{\em Case 1 - $f_i$ is strict}: Then $\llbracket Q\rrbracket{}_{S}(\vec f) = \llbracket Q\rrbracket{}_{S \cup \{i\}}(\vec f)$ by construction, and the case follows from the induction hypothesis.
\newline
{\em Case 2 -  $f_i$ is not strict}: Then it does not matter whatever we rewrite $Q$ to, and by the induction hypothesis $\llbracket P_a[S]\rrbracket{}(\vec f) = \llbracket P_a\rrbracket{}_{S}(\vec f)$, so the conclusion holds.
\item[2.2.] This case is trivial.
\item[2.3.] In this case we use that when $f_i$ is strict, then $\llbracket P\rrbracket{}(\vec f) =
  \llbracket \hat P\rrbracket{}(\vec f)$, and the rest follows by the induction hypothesis. \qedhere
\end{itemize}
\end{proof}
\begin{defi} Let $P$ be a sequential procedure. We say that $P$ is in \emph{$S$-normal form} if $P_S = P$.\end{defi}

\noindent A finite sequential procedure in $S$-normal form will have the cardinality of \[\{1, \ldots , n\}\setminus S\] as a left bound.

\section{The main Theorem}\label{Sec3}
In this section, we will make use of some basic facts about the finite sequential functionals proved in e.g. \cite{L-N,N-S}. One fact is that any sequential functional will be the least upper bound of a sequence of finite ones, in the sense that the graph is the union of the graphs of the approximations.  Another fact is that if two finite sequential functionals $F$ and $G$ are bounded by a sequential functional $H$, then there is a sequential least upper bound $F \sqcup G$ of the two. This is not necessarily the set theoretical union of $F$ and $G$:
\begin{exa} Let $F$ and $G$ of type $(\iota \rightarrow \iota) \rightarrow \iota$ be defined by
\begin{itemize}
\item $F(f) = 0$ if $f(0) = 0$
\item $G(f) = 0$ if $f(1) = 0$
\end{itemize}
Then the least sequential upper bound will be the constant zero, while the least set theoretical upper bound will just be
\begin{center}$H(f) = 0$  if $f(0) = 0$ or $f(1) = 0.$
\end{center}
The point is that $H(f)$ is defined when $f(0) = 0$ and $f(1) = \bot$ and when $f(0) = \bot$ and $f(1) = 0$ and a sequential procedure cannot handle this.\end{exa}
\subsection{Theorem and Key Lemma}

This section will be devoted to the proof of the following
\begin{thm}\label{Main}{\em  Let $\{F_k\}_{k \in \N}$ be an increasing sequence of sequential functionals of type $(\iota \rightarrow \iota)^n \rightarrow \iota$.
Then the least upper bound $$\bigsqcup_{k \in \N}F_k,$$ seen as a function, will be sequential.}
\end{thm}
If we do not allow for non-strict arguments at all, this theorem is part of the folklore, and is indeed covered by the more general result due to Plotkin mentioned in the introduction. The theorem will be a special case of the Key Lemma, see Lemma \ref{KL}. The Key Lemma will be proved by induction, and the mentioned folklore result can be seen as the base case.

 Even in the simplest case, for $n = 1$, there will be an increasing sequence $\{F_k\}_{k \in \N}$ of finite sequential functionals of type $(\iota \rightarrow \iota) \rightarrow \iota$ that can be enumerated in a primitive recursive way such that the least upper bound has the complexity of the Turing jump:
 \begin{exa}
 Let $A = \bigcup_{k \in \N} A_k$ be a subset of $\N$ where $A$ is Turing equivalent to the Turing jump, each $A_k$ is finite with $A_k \subseteq A_{k+1}$ for each $k$ and such that the relation $n \in A_k$ is primitive recursive.

 Let $F_k(f) = 0$ if $f(0) \in A_k$ or if both $f(0) \leq  k$ and $f(1) \leq k$. The least upper bound $F$ will be a total functional of pure type 2. Even if we allow for repetitions of queries, there will be no computable sequential procedure for $F$. For any such procedure we will have that $f(0) \not \in A$ if and only if $f(1)$ is eventually called for. This means that the complement of $A$ is computable in any procedure for $F$.
 \end{exa} Thus it is impossible to give a constructive proof  of our main result. Arithmetical comprehension will suffice, but we will not put stress on pointing out where non-constructive arguments are used.

Throughout Section \ref{Sec3} we will let $\{F_k\}_{k \in \N}$ be a $\sqsubseteq$ - increasing sequence of finite sequential functionals and let $F = \bigsqcup_{k \in \N}F_k$ be the pointwise least upper bound. As noted in the start of the section, this suffices. We will use the terminology of the previous section. 
\begin{lem}\label{lemma3.2}{\em Without loss of generality, we may assume that $F_k\in \F_k$ for each $k \in \N$ and that $F_k = \pi_k(F_l)$ when $k \leq l$. }\end{lem}
\proof
Let $$G_k = \bigsqcup_{k \leq l}\pi_k(F_l).$$ 
Since $\F_k$ is finite, these least upper bounds exist as sequential functionals. Clearly the sequences $\{G_k\}_{k \in \N}$ and $\{F_k\}_{k \in \N}$ will have the same least upper bounds. \qed

From now on we will assume that the sequence we consider satisfies the conclusion of this lemma.
We will prove the following strengthening of Theorem \ref{Main} by reversed induction on the size of $S$:
\begin{lem}[Key Lemma]\label{KL}{\em 
$$\bigsqcup_{k \in \N}(F_k)_S$$ is sequential.}
\end{lem}
\noindent The proof of the Key Lemma will take what remains of this section.
\subsection{Matching paths}
We have to develop some notation and prove some partial results before we are able to give the proof of Lemma \ref{KL}. The argument will be the same for all $S$, except that for $S = \{1 , \ldots , n\}$ we will never be in a situation where we have to appeal to the induction hypothesis.

\begin{samepage}
\begin{defi}\leavevmode
\begin{enumerate}[label=\alph*)]
\item A sequential functional $F$ is an \emph{$S$-functional} if $F_S = F$.
\item An \emph{$S$-query} is a pair $(i,G)$ where $1 \leq i \leq n$ and
\begin{itemize}
\item $G$ is a constant $a \in \N$ if $i \in S$
\item $G$ is an $S \cup \{i\}$-functional if $i \not \in S$
\end{itemize}
\item An \emph{$S$-pre-path} is a finite sequence
$$(i_1,G_1,b_1), \ldots , (i_m,G_m,b_m)$$
where each $(i_j,G_j)$ is an $S$-query and each $b_j \in \N$.
\item Let $\vec f = (f_1 , \ldots , f_n)$ be a sequence of functions in $\N_\bot \rightarrow \N_\bot$ and let $$\theta = (i_1,G_1,b_1), \ldots , (i_m,G_m,b_m)$$ be an $S$-pre-path.
We say that $\vec f$ and $\theta$ \emph{match}  if $f_{i_j}(G_j(\vec f)) = b_j$ for each $j \leq m$.
\item An \emph{$S$-path} is an $S$-pre-path that is matching at least one sequence $\vec f$.
\end{enumerate}
\end{defi}
\end{samepage}

\noindent An $S$-path $ (i_1,G_1,b_1), \ldots , (i_m,G_m,b_m)$ will satisfy the following consistency condition: If $i_j = i_l$ and $G_j$ and $G_l$ are consistent, then $b_j = b_l$. This will not necessarily hold for pre-paths. If $P$ is a finite sequential procedure in $S$-normal form and $$\theta = (i_1,G_1,b_1), \ldots , (i_m,G_m,b_m)$$ is an $S$-path, we may \emph{evaluate} $P$ on $\theta$ as follows:
\begin{itemize}
\item If $P$ is a constant $b$ we output $b$ directly.
\item Let $P = $ case $x_i(Q)$ of $\{b \Rightarrow P_b\}$.
If there is a $j \leq m$ such that $i_j = i$ and $G_j \sqsubseteq \llbracket Q\rrbracket{}$ we let $b_j$ be the first intermediate value of the evaluation, and continue with evaluating $P_{b_j}$ on $\theta$.
If there is no such $j \leq m$, the evaluation stops with output $\bot$.
We write $P(\theta)$ for the output of this process.
\end{itemize}
The paths will be forms of approximations to input sequences in the sense that they will carry the information about an input sequence we may have obtained after a partial evaluation of a sequential procedure. 
\begin{rem} In the definition of how to evaluate a finite procedure on a path, we almost implicitly assume that the functionals $G_j$ in the path are themselves finite, or at least bounded by finite objects, since we require that $G_j \sqsubseteq \llbracket Q\rrbracket{}$. We are going to  make use of paths where the $G$'s are infinite sequential functionals obtained by the induction hypothesis. We will have to work carefully with approximations in order still to be able to evaluate a procedure on a path in a reasonable sense.\end{rem}

It is clear that if $\vec f$ matches $\theta = (i_1,G_1,b_1), \ldots , (i_m,G_m,b_m)$ and $P$ is a sequential procedure in $S$-normal form that evaluated on $\theta$ yields a value $c$, then $\llbracket P\rrbracket{}(\vec f) = c$. As a consequence, we see that if we evaluate two equivalent procedures in $S$-normal forms on the same $S$-path, the results will be consistent, meaning that if they both terminate, the output will be the same. We do not think that the outputs have to be equal, i.e. one may terminate and the other not, but have not produced an example.
\begin{defi}
Let $F$ be an $S$-functional, $\theta = (i_1,G_1,b_1), \ldots , (i_m,G_m,b_m)$ an $S$-path.
An $S$-query $(i,G)$ is \emph{critical} for $F$ at $\theta$ if for every sequence $\vec f$ matching $\theta$  for which $F(\vec f)\!\!\downarrow$ we have that $f_i(G(\vec f)) \in \N$.
\end{defi}

If  $\vec f$ is matching $\theta$ and $f_i(G(\vec f)) \in \N$, then $\vec f$ will match $\theta$ extended with $(i,G,f_i(G(\vec f)))$. Our strategy will be to use $S$-queries critical for $F$ to build up paths matching $\vec f$ in a deterministic way  until, when $F(\vec f) \in \N$, we have sufficient information about $\vec f$ to know the value of $F(\vec f)$.
\begin{lem}\label{lemma3.6}{\em 
Let $F = \llbracket P\rrbracket{}$ be a finite $S$-functional, where $P$ is a procedure in $S$-normal form. Let $\theta$ be a path such that $P$ evaluated on $\theta$ does not yield a value in $\N$.
Assume that there is some $\vec f$ matching $\theta$ such that $F(\vec f) \in \N$.
Then there is an $S$-query $(i,G)$ that is critical for $F$ at $\theta$, and moreover, such that there is no $j \leq m$ with $i_j = i$ and $G_j \sqsubseteq G$.}

\end{lem}
\proof
For any $\vec f$ that matches $\theta$, the evaluation of $P$ on $\vec f$ will follow the same path through the procedure tree  as the evaluation of $P$ on $\theta$, until the latter comes to a halt. This must be because the evaluation reached a subprocedure \begin{center}
$P' = $ case $x_i(Q)$ of $\{a \Rightarrow P'_a\}$ \end{center} while there is no $j \leq m$ such that $i = i_j$ and $G_j \sqsubseteq \llbracket Q\rrbracket{}$. For the evaluation on $\vec f$ to proceed, we must have that $f(\llbracket Q\rrbracket{}(\vec f)) \in \N$.
Let $G = \llbracket Q\rrbracket{}$. Then $(i,G)$ will be critical for $F$ at $\theta$ and satisfy the extra requirement.
\qed

\subsection{Projections and preservation}
One problem with the concepts we have defined so far is that if $P$ is a procedure that we may evaluate on a path $\theta$, and we extend the $G_j$'s in the path to some $G'_j$'s, then the evaluation may not work anymore. This is not so strange as it may seem. If we extend the items of a path, the path will match more function sequences $\vec f$, also some for which the original procedure does not terminate. However, it is in order to handle the technical problem induced by this phenomenon that we assume that our increasing sequence $\{F_k\}_{k \in \N}$ has the properties of Lemma \ref{lemma3.2}. We will take the liberty to extend the use of the term $\pi_k$ for the projection to $\F_k$ for all types in question, in particular to $\I_k$ and $\I^n_k$ defined by:
\begin{defi}\leavevmode
\begin{enumerate}[label=\alph*)]
\item Let $\I$ be the functions of type $\iota \rightarrow \iota$, and let $\I_k$ be the corresponding set of functions in $\{\bot,0,\ldots,k\} \rightarrow \{\bot,0, \ldots , k\}$.
\item If $(i,G)$ is an $S$-query, we let $\pi_k(i,G) = (i,\pi_k(G))$.
\item If $\theta = (i_1,G_1,b_1), \ldots , (i_m,G_m,b_m)$ is a path, we let $$\pi_k(\theta) = (i_1,\pi_k(G_1),b_1), \ldots , (i_m,\pi_k(G_m),b_m).$$
\end{enumerate}\end{defi}
\noindent Notice that $\pi_k(\theta)$ is not necessarily a path, since it does not have to match any $\vec f$ if too much information is removed.

\begin{lem}\label{lemma3.8}\leavevmode
\begin{enumerate}[label=\alph*)]
\item If $F \in \F$ and $\vec f \in \I^n_k$, then $\pi_k(F(\vec f)) = \pi_k(F)(\vec f)$.
\item If $F \in \F_k$ and $\vec f \in \I^n$, then $F(\vec f) = F(\pi_k(\vec f))$.
\item If $\vec f \in \I^n_k$ and $\theta$ is a path matching $\vec f$, then $\pi_k(\theta)$ is a path matching $\vec f$.
\end{enumerate}
\end{lem}
\proof
a) and b) follows from the iterative construction of the projection $\pi_k$ over the types. The definition of matching requires that for any $(i,G,b)$ in $\theta$ we must have that $f_i(G(\vec f)) = b$. For this to hold, we must either have that $f_i$ is constant $b \leq k$ or that $G(\vec f) \leq k$. In the first case, the claim is trivial, and in the second case we may use $a)$.
This establishes c). \qed
We will extend the $\pi_k$ also to procedures:
\begin{defi} If $P$ is a procedure, we let $\pi_k(P)$ be the procedure we obtain by the following top-down recursive transformation:
\begin{itemize}
\item[-] $\pi_k(a) = a$ if $a \leq k$ and $\bot$ if $k < a$
\item[-] If $P$ =  case $x_i(Q)$ of $\{a \Rightarrow P_a\}_{a \in \N}$, we let $\pi_k(P) = $ case $x_i(\pi_k(Q))$ of $\{a \Rightarrow \pi_k(P_a)\}_{a \leq k}$
\end{itemize}\end{defi}

\noindent It is easy to see that $\pi_k(\llbracket P\rrbracket{}) = \llbracket \pi_k(P)\rrbracket{}$, because any evaluation of $P$ with an input from $\I^n_k$ can only utilise values bounded by $k$ in the evaluation.

\begin{lem}\label{lemma3.11}{\em  Let $\theta$ be a path, $F$ a functional  and  $(i,G)$ be critical for $F$ at $\theta$.
Then $(i,\pi_k(G))$ is critical for $\pi_k(F)$ at $\pi_k(\theta)$.}
\end{lem}
\proof
Let $\vec f \in \I^n_k$ be such that $\pi_k(\theta)$ matches $\vec f$ and such that $\pi_k(F)(\vec f)\!\!\downarrow$.
Then $\theta$ matches $\vec f$ and $F(\vec f)\!\!\downarrow$, so $f_i(G(\vec f))\!\!\downarrow$. By Lemma \ref{lemma3.8} then, $f_i(\pi_k(G)(\vec f))\!\!\downarrow$. \qed

\subsection{Proof of the Key Lemma}
Recall that we will prove this lemma by reversed induction on the cardinality of $S$.
We will now let the sequence $\{F_k\}_{k \in \N}$ be given, satisfying the assumption of Lemma \ref{lemma3.2}, and we let $\vec f \in \I^n$ be given. We have that $F_S = \bigsqcup_{k \in \N} (F_k)_S$, and we will devise a sequential procedure for evaluating $F_S(\vec f)$. To this end, we will construct a path  matching $\vec f$ step by step until we  have enough information about $\vec f$ to know the value of $F_S(\vec f)$. We will prove that if $F_S(\vec f)\!\!\downarrow$, then our procedure applied to $\vec f$ will come to a halt, and that whenever it halts, the conclusion will be the right one for the given input. \smallskip

For each $k \in \N$, let $P_k$ be in $S$-normal form such that $(F_k)_S = \llbracket P_k\rrbracket{}$.
We will describe one step in the construction:

Let $\theta = (i_1,G_1,b_1), \ldots , (i_m,G_m,b_m)$ be the path matching $\vec f$ constructed after step $m$. 
\begin{enumerate}
\item If there is some $k$ such that $P_k(\pi_k(\theta))\!\!\downarrow$ we know that $F_S(\vec f) = P_k(\pi_k(\theta))$ and we output this value. 
\item If there is no $\vec g$ such that $\vec g$ matches $\theta$ and $F_S(\vec g)\!\!\downarrow$, we know that $F_S(\vec f) = \bot$, and our procedure applied to $\vec f$ comes to a halt with output $\bot$.
\item Otherwise, let $k$ be minimal such that for some $\vec g \in \I^n_k$ matching $\theta$ we have that $(F_k)_S(\vec g)$ terminates. Select one such $\vec g$. Notice that we can split between the cases 1., 2. and 3. on the basis of $\theta$, and that the values of $k$ and $\vec g$ in this case only depend on $\theta$, so we do not need to know $\vec f$ in order to know what to do.
\end{enumerate}
In case of alternative (3), we will see how to extend $\theta$ to a longer path matching $\vec f$:

Let $l \geq k$. Then $P_l(\vec g)\!\!\downarrow$ while $P_l(\pi_l(\theta))= \bot$. By Lemma \ref{lemma3.6}  there will be an $S$-query $(i_l,H_l)$ that is critical for $(F_l)_S$ at $\pi_l(\theta)$ such that there is no $j \leq m$ with $i_l = j$ and $\pi_l(G_j) \sqsubseteq H_l$. Then, by Lemma \ref{lemma3.11}, $(i_l,\pi_k(H_l))$ will be critical for $\pi_k(F_k)$ at $\pi_k(\theta)$. Since $g_{i_l}(H_l(\vec g))\!\!\downarrow$ and $\vec g \in \I^n_k$, we have that $g_{i_l}(\pi_k(H_l)(\vec g))\!\!\downarrow$. Since $\pi_k(F_l) = F_k$ we also have that $\pi_k(P_l)$ is a procedure for $F_k$, and the evaluation of $\pi_k(P_l)(\vec g)$ will match the evaluation of $\pi_k(P_l)(\pi_k(\theta))$ for as long as the latter evaluation goes. It follows that $(i_l,\pi_k(H_l))$ will be the query unanswered by $\pi_k(\theta)$ that we get when proving Lemma \ref{lemma3.6} for $\pi_k(P_l)$, $\pi_k(\theta)$ and $\vec g$.  These considerations contain a proof of the following claim:
\smallskip
\begin{description}
\item[Claim]
 \emph{ For each $l \geq k$, let $X_l$ be the set consisting of pairs $(i,H)$ that are critical $S$-queries for $(F_l)_S$ at $\pi_l(\theta)$ and such that for no $j$ we have that $i = j$ and $\pi_k(G_j) \sqsubseteq \pi_k(H)$. Then $X_l$ is nonempty.}
 \end{description}

\noindent If $k \leq l_1 \leq l_2$, then $\pi_{l_1}$ will map $X_{l_2}$ into $X_{l_1}$, and each $X_l$ is finite.  We use K\"onig's lemma  to see  that there will be one $i \leq n$ and one increasing sequence $\{H_l\}_{l \geq k}$ of $S \cup \{i\}$-functionals, with each  $(i,H_l) \in X_l$. We select one such sequence (depending on $\theta$, but not on $\vec f$).

If $i \in S$ we actually have that each $H_l$'s is a   constant, and that they are all the same. In this case, we let $H$ be this constant. If $i \not \in S$, we use the induction hypothesis, and let $H = \bigsqcup_{l \geq k}H_k$ be sequential.

If $f_i(H(\vec f)) = \bot$ our procedure for $F_S(\vec f)$ will not terminate, which is as it should be since $(i,H)$ is critical at $\theta$, so $F_S(\vec f)$ will not terminate either.  If $f_i(H(\vec f)) = b \in \N$, we add $(i,H,b)$ to our path and continue the process. This is the only place where we actually use $\vec f$, and clearly, the new path will also match $\vec f$.
\smallskip

This ends one full step in the construction of the path matching $\vec f$.
Clearly, if this process on $\vec f$ leads to a conclusion after finitely many steps, the output in $\N_\bot$ will coincide with $F_S(\vec f)$. It remains to prove that if $F_S(\vec f)\!\!\downarrow$, then the process will terminate.

Let $F_S(\vec f)\!\!\downarrow$ and let $k_0$ be minimal such that $F_S(\pi_{k_0}(\vec f))\!\!\downarrow$. Without loss of generality we may assume that $\vec f \in \I^n_{k_0}$, and then $(F_{k_0})_S(\vec f)\!\!\downarrow$.

When we apply our strategy to the input $\vec f$, we will build up a path matching $\vec f$, and as long as the $S$-normal procedure for $(F_{k_0})_S$ does not terminate on the path built up, the $k$ of the construction will be bounded by $k_0$. But as long as this is the case, we are properly extending the path even when projected to the $k_0$-level. This level is finite, so this cannot go on forever. The conclusion is that our process, when applied to $\vec f$, must terminate.

What remains is to summarise how we construct a sequential procedure from the $S$-paths.
The tree of paths depends only on $S$, $F$, the sequence of $F_k$'s and the choice of $\vec g$ in each step, and which branch, or path, we follow in that tree depends deterministically on $\vec f$. Each item $(i,G,b)$ in one of the $S$-paths constructed will correspond to an instruction
\begin{center}case $x_i(G)$ of $b \Rightarrow \cdot$\end{center}
in the sequential procedure for $F_S$ that we actually construct. It is easy to
see that this top-down translation of the tree of $S$-paths to a sequential
procedure will work. $\hfill\qEd$

%
\section{Conclusions and discussion}

In this note we have not just completed the characterisation of when the sequential functionals of a type $\sigma$ form a \emph{dcpo}. As a bonus, we have proved that the sequential functionals of the types $(\iota \rightarrow \iota)^n \rightarrow \iota$ are what is called ``left bounded" in \cite{L-N}, and with a uniform bound on the left nesting needed. One key aspect of the examples in \cite{DN1,N-S} is that the sequence without a sequential upper bound in both cases will have increasing left depth, and  in the limit there has to be an evaluation with an infinite left nesting of sub-procedure calls. It is then natural to ask if the substructure of the sequential functionals of any type with a fixed bound on the left nesting will form a \emph{dcpo}. This seems likely, but since these classes in general do not seem to possess other interesting closure properties, a result along these lines will be of limited interest.

In Section 7 of \cite{N-S} some further open problems were listed, and we rephrase them here. It is still not known exactly when a finite sequential functional is the least upper bound of some non-trivial increasing sequence of sequential functionals. We know that a finite set of finite sequential functionals with a sequential upper bound will have a least sequential upper bound, but this is unknown for finite classes of sequential functionals in general. We may also ask if a stronger property may hold, will every finite bounded set of sequential functionals have a sequential least upper bound?  These problems may be worthwhile to look into, not because the results will be of particular interest in themselves, but because solving the problems may lead to a better understanding of  the sequential functionals.
\section*{Acknowledgement}I am grateful to two anonymous referees for pointing out typos and language errors, and for helpful suggestions concerning the exposition.

\end{document}